\newif\ifabstr
\abstrfalse 

\ifabstr
\documentclass[11pt, a4paper]{article}
\usepackage[a4paper,includeheadfoot,margin=2cm]{geometry}
\else
\documentclass[11pt, a4paper]{article}
\usepackage{fullpage}
\fi

\usepackage[latin1]{inputenc}

\usepackage{amsmath}
\usepackage{amssymb,amsthm,graphicx}
\usepackage[dvips]{epsfig}
\usepackage{amsfonts}
\usepackage{xcolor}
\usepackage{hyperref}
\usepackage{enumerate}
\usepackage[normalem]{ulem} 
\usepackage{authblk}

\newtheorem{theorem}{Theorem}
\newtheorem{proposition}[theorem]{Proposition}

\newtheorem{observation}[theorem]{Observation}

\newtheorem{lemma}[theorem]{Lemma}

\newtheorem{claim}{Claim}[theorem]
\theoremstyle{remark}

\theoremstyle{definition}
\newtheorem{definition}[theorem]{Definition}

\newcommand{\R}{\mathbb{R}}

\newcommand{\Z}{\mathbb{Z}}

\definecolor{darkgreen}{rgb}{0,.5,0}

\newcounter{sideremark}

\DeclareMathOperator{\sd}{sd}

\DeclareMathOperator{\im}{im}

\DeclareMathOperator{\cat}{cat}
\DeclareMathOperator{\gcat}{gcat}
\DeclareMathOperator{\plgcat}{plgcat}

\title{NP-hardness of computing PL geometric category in dimension 2
\ifabstr\else
\thanks{
This work was supported by GA\v{C}R grant 22-19073S.
MS acknowledges support by the project ``Grant Schemes at CU'' (reg. no. CZ.02.2.69/0.0/0.0/19\_073/0016935).
}\fi}
\ifabstr
\author{}
\else
\author{Michael Skotnica}
\author{Martin Tancer}

\affil{Department of Applied Mathematics, Faculty of Mathematics and Physics,
Charles~University, Prague, Czech Republic}
\fi

\date{}

\begin{document}

\maketitle

\begin{abstract}
  The PL geometric category of a polyhedron $P$, denoted $\plgcat(P)$, is
  a combinatorial notion which provides a natural upper bound for the
  Lusternik--Schnirelmann category
  and it is defined as the minimum number of PL collapsible subpolyhedra
  of $P$ that cover $P$.  In dimension 2 the PL geometric category is at most~3.

  It is easy to characterize/recognize $2$-polyhedra $P$ with $\plgcat(P) =
  1$. Borghini provided a partial characterization of $2$-polyhedra with
  $\plgcat(P) = 2$. We complement his result by showing that it is NP-hard to decide
  whether $\plgcat(P)\leq 2$. Therefore, we should not expect much more than a
  partial characterization, at least in algorithmic sense. Our reduction is
  based on the observation that 2-dimensional polyhedra $P$ admitting a shellable
  subdivision satisfy $\plgcat(P) \leq 2$ and a (nontrivial) modification of
  the reduction of Goaoc, Pat\'{a}k, Pat\'{a}kov\'{a}, Tancer and Wagner
  showing that shellability of $2$-complexes is NP-hard.
\end{abstract}

\section{Introduction} \label{sec:introduction}
%

An important notion in homotopy theory is the Lusternik--Schnirelmann category
(LS category) 
of a topological space. This notion is important not only as a purely
mathematical object (see, e.g., the book~\cite{cornea-lupton-oprea-tanre03})
but also in computer science as it is closely related to the topological
complexity of motion planning; see, e.g,~\cite{farber03, farber04, farber-mescher20}.

The \emph{LS category}, $\cat(X)$, of a topological space $X$ is the smallest $n$ (if it exists)
such that $X$ can be covered by $n$ open sets so that the inclusion of
each of the open sets is nullhomotopic in $X$. One difficulty when
working with the LS category is that it is often hard to determine. For
example, determining whether $\cat(X) = 1$ is equivalent to contractibility
of $X$. This is known to be undecidable if $X$ is a simplicial complex of
dimension at least $4$; see~\cite[\S 10]{volodin-kuznecov-fomenko74}
and~\cite[Appendix]{tancer16} while it is an open problem whether this is
decidable for simplicial complexes of dimension $2$.\footnote{Via tools
in~\cite{hog-metzler-sieradski93} (using the exercise on page
8) decidability of this problem is equivalent to determining whether a given
balanced presentation of a group presents a trivial group. In this form, the
problem is mentioned for example in~\cite{baumslag-myasnikov-shpilrain02}.}
(In this paper we are mainly interested in complexes of dimension $2$.)

In order to bound the LS category from above we can use some closely related
notions. One of them is the geometric category, $\gcat(X)$, which requires that
the open sets covering $X$ are already contractible.
(For more details see again~\cite{cornea-lupton-oprea-tanre03}.)
If $X$ is a polyhedron,
this is equivalent to finding the minimum number of subpolyhedra covering $X$
each of which is contractible. This may make estimating $\gcat(X)$ sometimes
easier. However, determining whether $\gcat(X) = 1$ is still equivalent to
contractibility of $X$.

Next step in this direction has been done by Borghini~\cite{borghini20} who
introduced PL geometric category $\plgcat(P)$ of a compact (connected)
polyhedron $P$. It is the minimum number of PL collapsible subpolyhedra of $P$
that cover $P$. (See Section~\ref{s:prelim} for the precise
definition of PL collapsibility.) In this case determining whether $\plgcat(P) = 1$ is equivalent to asking whether $P$ is PL collapsible. 
At least for $2$-complexes this is a significant improvement as PL
collapsibility of $2$-complexes is a purely combinatorial notion which is easy to check. Indeed, it is not hard to derive from known results that this is a polynomially checkable criterion (by performing the collapses greedily on an arbitrary triangulation).


\begin{proposition}
\label{p:polynomial}
Given a $2$-dimensional triangulated polyhedron $P$, it can be checked in
  polynomial time whether $\plgcat(P) = 1$.
\end{proposition}

Borghini further proved~\cite{borghini20} that a connected $d$-dimensional polyhedron has
PL geometric category at most $d+1$. For connected $2$-polyhedra, this means
that the only options are $1$, $2$, or $3$. One of the main aims
in~\cite{borghini20} is to provide a partial characterization of polyhedra $P$
with $\plgcat(P) \leq 2$ (which we do not reproduce here). All these
positive results suggest that 
determining $\plgcat(P)$ could be easy for $2$-polyhedra.
In particular, one should be curious whether it is possible to extend Borghini's results
to a \emph{full characterization} that would distinguish $2$-polyhedra with PL geometric
category equal to $2$ from those for which it equals $3$.

We will show that this is essentially impossible, at least for
an efficiently algorithmically checkable characterization. In technical terms, we show
that determining whether $\plgcat(P) \leq 2$ is NP-hard. Let us point out that NP-hard
problems are believed not to be solvable in polynomial time. (This is
equivalent to the standard conjecture P$\neq$NP in theory of computation.
NP-hardness is discussed in a bit more detail in Section~\ref{s:prelim}.)


\begin{theorem}
\label{t:main}
  Given a $2$-dimensional triangulated polyhedron $P$, it is NP-hard to decide
  whether $\plgcat(P) \leq 2$.
\end{theorem}

We should also point out that we actually do not know whether recognition of
triangulated polyhedra with $\plgcat(P) \leq 2$ belongs to the class NP (not
even whether it is decidable). This could be certified by two subpolyhedra
witnessing $\plgcat(P) \leq 2$ but we do not know whether we can bound their
sizes.

A useful step towards our proof is that we observe a relation between $\plgcat(P) \leq 2$ and
shellability (of some triangulation) of $P$.
(Shellability will be discussed in
detail in Section~\ref{s:prelim}.)

\begin{proposition}
\label{p:shell}
  If a $2$-dimensional polyhedron $P$ admits a (pure) shellable triangulation, then
  $\plgcat(P) \leq 2$.
\end{proposition}

It has been shown by Goaoc, Pat\'ak, Pat\'akov\'a, Tancer and
Wagner~\cite{goaoc-patak-patakova-tancer-wagner19} that shellability is NP-hard
already for $2$-dimensional simplicial complexes. In addition, the reduction
in~\cite{goaoc-patak-patakova-tancer-wagner19} is quite resistant with respect
to subdivisions. Thus, we could hope to prove Theorem~\ref{t:main} in the
following way: Consider a complex $K$ that appears in the reduction
in~\cite{goaoc-patak-patakova-tancer-wagner19}. If $K$ is shellable, then
$\plgcat(|K|) \leq 2$  by Proposition~\ref{p:shell} ($|K|$ stands for the
polyhedron of $K$). If we were able to show the other implication: `if $K$ is
not shellable, then $\plgcat(|K|) = 3$', we would immediately get a proof of
Theorem~\ref{t:main}. Unfortunately, the other implication, stated this way, is
not true: with some more effort (which we do not do here), it could be shown
that every complex $K$ from the reduction
in~\cite{goaoc-patak-patakova-tancer-wagner19} satisfies $\plgcat(|K|) = 2$.
However, this problem can be circumvented. We construct certain enriched
complex $K^+$ (by attaching a torus in a suitable way to every triangle of
$K$---it may be slightly surprising that this indeed helps). It turns out that $\plgcat(|K^+|)$ stays $2$ for shellable $K$ but it
grows to $3$ for non-shellable $K$ (coming
from~\cite{goaoc-patak-patakova-tancer-wagner19}). This will prove
Theorem~\ref{t:main}.

We point out that Proposition~\ref{p:shell} as stated is not
really necessary in the proof of Theorem~\ref{t:main}. But we state it here as
it provides the motivation for our approach as well as it can be seen as a
complementary result to the results of Borghini~\cite{borghini20} providing
some sufficient (or necessary)
conditions for $\plgcat(P) \leq 2$.

We also point out that instead
of~\cite{goaoc-patak-patakova-tancer-wagner19}, it would be in principle
possible to use also a modification of reduction by Santamar\'{\i}a-Galvis and
Woodroofe~\cite{santamariagalvis-woodroofe21} where some of the gadgets are
slightly simplified. However, some intermediate steps
in~\cite{goaoc-patak-patakova-tancer-wagner19} are done via collapsibility thus
for purposes of this paper it is easier to adapt to the setting
in~\cite{goaoc-patak-patakova-tancer-wagner19}.

\paragraph{Organization.} Proposition~\ref{p:polynomial} is proved
in Section~\ref{s:collapsibility}; Theorem~\ref{t:main} is proved in
Section~\ref{s:NP-hardness} and Proposition~\ref{p:shell} is proved in
Section~\ref{s:prelim}.

\section{Preliminaries}
\label{s:prelim}

\paragraph{Simplicial complexes, polyhedra and subdivisions.} 
Although we assume that the reader is familiar with simplicial
complexes (abstract or geometric), we briefly recall these notions up to the
level we need in this paper.
Because we also want
to work with polyhedra, we will be using geometric simplicial complexes
(with a single exception that the input for any computational problem we
consider is the corresponding abstract simplicial complex).
That is, a \emph{simplicial complex} is
for us a collection of (geometric) simplices embedded in some $\R^m$ such that
two simplices intersect in a face of both of them; and a face of any simplex in the
complex belongs again to the complex. The \emph{dimension} of a simplex
is the number of its vertices minus one; the \emph{dimension} of a simplicial complex
is the maximum of the dimensions of simplices appearing in the complex.
When considering a simplicial
complex as input of an algorithmic problem, we switch to the corresponding
abstract simplicial complex. Roughly speaking, it records only the
combinatorial information which vertices form a simplex.
For more details on simplicial complexes,
we refer to textbooks such as~\cite{rourke-sanderson82, matousek03}.

We work with \emph{polyhedra} as defined in~\cite{rourke-sanderson82}. When we say `polyhedron' we always mean a compact
polyhedron. Because every compact polyhedron can be triangulated, an equivalent
definition is that a polyhedron is the underlying space $|K| := \bigcup_{\sigma
\in K} \sigma$ of some finite
simplicial complex $K$ (a.k.a. the \emph{polyhedron of $K$}).

A simplicial complex $K'$ is a \emph{subdivision} of a complex $K$ if
$|K'| = |K|$ and every $\sigma' \in K'$ is a subset of some $\sigma \in K$.
Given a subcomplex $L$ of $K$, then the \emph{subcomplex $L'$ of $K'$
corresponding to $L$ is the complex $L' := \{\sigma' \in K'\colon \sigma
\subseteq |L|\}$. 
}

\paragraph{Collapsibility and PL collapsibility.} 
Given a simplicial complex $K$, a face $\sigma \in K$ is \emph{free}, if it is contained
in a unique maximal face. A complex $K$ collapses to a subcomplex $K'$ by an elementary collapse, if $K'$ is obtained from $K$ by
removing a pair of faces $\{\sigma, \tau\}$ where $\dim \tau = \dim \sigma +
1$, $\sigma$ is free and $\tau$ is the maximal face containing $\sigma$.\footnote{Some authors allow more general elementary collapse removing
a face $\sigma$ and all faces containing it provided that $\sigma$ is contained
in a unique maximal face. This is only a cosmetic change in the resulting
notion of collapsible complex because this more general elementary collapse can
be emulated by a sequence of elementary collapses according to our definition.}
We also say in this case that $K$ collapses to $K'$ \emph{through $\sigma$}.
A simplicial complex $K$ 
\emph{collapses} to a subcomplex $L$, if there is a
sequence of elementary collapses starting collapsing $K$ gradually to $L$. A
complex $K$ is collapsible, if $K$ collapses to a point. Because
collapsibility preserves the homotopy type, collapsibility of a simplicial
complex is often understood as a combinatorial counterpart of the notion of
contractibility of a topological space. (In particular, collapsibility of a
complex implies contractibility of the corresponding polyhedron.)

A polyhedron $P$ is \emph{PL collapsible} if some triangulation of $P$ is a
collapsible simplicial complex. Similarly, a simplicial complex $K$ is \emph{PL
collapsible} if $|K|$ is a PL collapsible polyhedron. Here, we should point
out a certain subtlety in the definition of PL collapsible simplicial complex:
If $K$ is PL collapsible, then there is some triangulation $K'$ of $|K|$ which
is collapsible (in the simplicial sense). This triangulation $K'$ needn't be a
priori a subdivision of $K$. However, by~\cite[Theorem~2.4]{hudson69} we may
assume that $K'$ actually is a subdivision of $K$. This also affects our
earlier definition of $\plgcat(|K|)$. We get $\plgcat(|K|) \leq k$ if and only
if some subdivision of $K$ can be covered by $k$ collapsible subcomplexes while
it does not matter with which triangulation of $|K|$ we start.

In general, collapsibility and PL
collapsibility of a simplicial complex differ because PL collapsibility allows
an arbitrarily fine subdivision before starting the collapses. In this paper,
we need both and we carefully distinguish these two notions.

\paragraph{Shellability.}
A simplicial complex $K$ is \emph{pure} if all its (inclusion-wise) maximal faces have the same
dimension. A \emph{shelling} of a pure complex $K$ is an ordering of all its
maximal faces into a sequence $\vartheta_1, \dots, \vartheta_m$ such that for
every $k \in \{2, \dots, m\}$ the subcomplex of $K$ with the underlying space
$(\bigcup_{i = 1}^{k-1} \vartheta_i) \cap \vartheta_k$ is pure and $(\dim
\vartheta_k - 1)$-dimensional. (Here we use the notation for geometric
simplicial complexes, thus $\vartheta_1, \dots, \vartheta_m$ are actual
geometric simplices.) A complex $K$ is \emph{shellable} if it admits a
shelling.

There are some similarities between collapsible and shellable simplicial
complexes. However, in general, these two notions differ. For example,
on the one hand
a collapsible complex is always contractible as an elementary collapse keeps the
homotopy type but shellable complexes
need not be contractible. On the other hand,
the union of two triangles meeting in a single vertex is a
complex which is collapsible but not shellable. The following description
of $2$-complexes admitting a shellable subdivision has been given by
Hachimori~\cite{hachimori08}.

\begin{theorem}[\cite{hachimori08}]
\label{t:hachimori}
 Let $K$ be a $2$-dimensional simplicial complex. Then the following statements are equivalent: 
\begin{enumerate}[(i)]
  \item The complex $K$ has a shellable subdivision.
  \item The second barycentric subdivision $\sd^2 K$ is shellable.
  \item The link of each vertex of $K$ is connected
    and $K$ becomes collapsible after removing $\tilde \chi (K)$ triangles
    where $\tilde \chi$ denotes the reduced Euler characteristic.
\end{enumerate}
\end{theorem}

Hachimori's theorem easily implies Proposition~\ref{p:shell}:

\begin{proof}[Proof of Proposition~\ref{p:shell}]
  Let $K$ be a pure shellable triangulation of $P$. By 
  Theorem~\ref{t:hachimori} there is a list of triangles $\tau_1, \dots,
  \tau_\ell$ such that the resulting complex $K'$ is collapsible after removing
  these triangles. Now we build an auxiliary complex $L$ from $K$ by
  subdividing each of the triangles $\tau_1, \dots, \tau_\ell$ as in
  Figure~\ref{fig:subdivision_of_F}. We also build a complex $L'$ by removing
  the middle triangle $\tau'_i$ from each subdivided $\tau_i$ in $L$. The complex $K'$ is
  a subcomplex of $L'$ and it is not hard to see that $L'$ collapses to $K'$.
  Hence $L'$ is collapsible as well. Then $|L'|$ is one of the two collapsible
  polyhedra covering $P$. The second polyhedron is obtained by taking the union
  of $\tau'_i$ and connecting them along the $1$-skeleton of $L$ so that the
  resulting complex is collapsible (the connection along the $1$-skeleton of
  $L$ can be, for example, obtained so that we pick two edges in each triangle
  and then we extend this forest to a spanning tree).
\end{proof}

\begin{figure}
  \centering
  \includegraphics[scale=1]{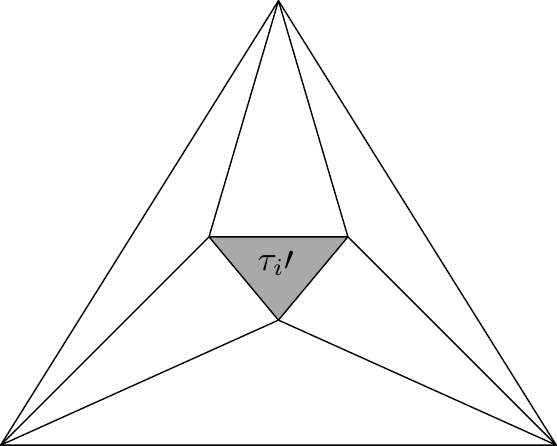}
  \caption{A subdivision of triangle $\tau_i$ into seven parts
    from the proof of Proposition~\ref{p:shell}.
  }
  \label{fig:subdivision_of_F}
\end{figure}

\paragraph{Homology.} In our auxiliary computations, we will often need homology groups,
including the exact sequence for pairs, the Mayer-Vietoris exact sequence and
the Lefschetz duality. In general, we refer to the literature such
as~\cite{hatcher02, munkres84} for details (in case of Lefschetz duality, we
will recall its statement when used).

In all our computations, we work with homology with $\Z_2$-coefficients. 
When working with simplicial complexes, we use simplicial homology. In
particular, when we speak of $k$-chains, then we can identify a $k$-chain with
a collection of $k$-simplices (in its support).
(Similarly, a $k$-cycle is such a collection with trivial boundary; i.e., each
$(k-1)$-simplex is in an even number of $k$-simplices of the cycle.) 
In case of polyhedra, we use
singular homology. However, we of course implicitly use that the simplicial and
singular homology groups are (naturally) isomorphic (for a simplicial complex
and its polyhedron).

%
%

\paragraph{NP-hardness and satisfiability.}
Here we briefly overview a few notions from computational complexity we need in this paper.
For more details see, e.g.,~\cite[Chapter~2]{arora-barak09}.
 

  A decision problem belongs to the \emph{class NP} if an affirmative answer to it can be verified in polynomial time
  using a certificate of polynomial size.
  A decision problem $X$ is \emph{NP-hard} if for each problem $Y$ from the class NP there is a polynomial time reduction from $Y$ to $X$.
  More precisely, given an instance $q$ of
  the problem $Y$
  one can construct in polynomial time in the size of $q$ an instance $p$ of the problem $X$ such that the answer to $q$ is yes
  if and only if the answer to $p$ is yes.

  An important NP-hard problem is the so called $3$-satisfiability problem (it
  also belongs to NP).
  An input for the $3$-satisfiability problem is a $3$-CNF formula~$\phi$,
  that is, a boolean formula in conjunctive normal form
  where every clause contains exactly three literals.\footnote{A
  \emph{literal} is some variable $x$ or its negation $\neg x$; a \emph{clause}
  with $3$ literals is a (sub)formula of form $(\ell_1 \vee \ell_2 \vee \ell_3)$
  where $\ell_i$ are literals. A formula $\phi$ is in conjunctive normal form if
  it can be written as $\phi = c_1 \wedge c_2 \wedge \cdots \wedge c_m$ where
  $c_j$ are clauses. An example of a $3$-CNF formula is $(x \vee \neg y \vee z)
  \wedge (\neg x \vee \neg y \vee t)$.}
  The output is the answer whether the formula is satisfiable,
  that is, whether it is possible to assign the variables TRUE or FALSE so that the formula evaluates to TRUE in this
  assignment.

  It is well known that $3$-satisfiability is NP-hard. In order to show that
  another problem $X$ is NP-hard, it is sufficient to construct a polynomial
  time reduction from $3$-satisfiability to $X$.
    
  
  


\section{PL collapsibility of 2-complexes}
\label{s:collapsibility}

It is a folklore result going back at least to Lickorish (according
to~\cite{hog-metzler-sieradski93}) that simplicial $2$-complexes can be
collapsed greedily:

\begin{proposition}[see~{\cite[page~20]{hog-metzler-sieradski93}}
  or~{\cite[Lemma~1 + Corollary~1]{malgouyres-frances08}}]\label{p:greedy_collapses}
  Let $K$ be a collapsible $2$-complex. Assume that $K$ collapses to a subcomplex
  $L$. Then $L$
  is collapsible as well. In particular, it can be checked in
  polynomial time whether a simplicial $2$-complex is collapsible.
\end{proposition}

For PL collapsibility we can essentially deduce the same conclusion as for
collapsibility as soon as we observe that PL collapsibility of a $2$-complex
does not depend on the choice of the subdivision, which also might be a
folklore result.

\begin{lemma}\label{lem:collapsible_subdivision}
  Let $K$ be a simplicial complex of dimension at most $2$ and
  $K^\prime$ be a subdivision of $K$.
  Then $K$ is collapsible if and only if $K^\prime$ is collapsible.
\end{lemma}

In the proof of the lemma we use the following observation.

\begin{observation}\label{obs:collapsing_of_triangle}
   Let $\tau$ be a triangle with vertices $a,b,c$. Let $K'$ be an arbitrary
   subdivision of $\tau$.
   Then $K'$ collapses to the subcomplex $V'$ formed by the subdivision of the edges $ab$ and $bc$.
\end{observation}

\begin{proof}
We greedily perform collapses through free edges of $K'$ which are not in $V'$.
  Let $L'$ be the resulting complex. We observe that $L'$ contains no triangle.
  Indeed, every edge contained in some triangle of $L'$ is either an edge of
  $V'$ or it has to be contained in both neighboring triangles (otherwise we
  could continue with collapses). This means, because the dual graph of $K'$ is
  connected, that once there is a single triangle of $K'$ in $L'$, then $L'$
  contains all triangles of $K'$ which is a contradiction.

Thus, $L'$ contains no triangles and it has the same homotopy type as $K'$.
  That means that $L'$ is a tree. Now we greedily perform collapses of edges
  not in $V'$ through vertices of degree $1$. By essentially the same argument
  as above, only the edges of $V'$ remain (otherwise, we would find a cycle in
  $L'$).
\end{proof}

\begin{proof}[Proof of Lemma~\ref{lem:collapsible_subdivision}]
First, we show that if $K$ is collapsible, then $K'$ is collapsible by induction
    on the number of simplices of $K$ (the case of one vertex is trivial).
    Assume that $K_1$ arises from $K$ by the first elementary collapse in some collapsing of $K$.
    First, assume that it removes an edge $ac$ and a triangle $abc$. Perform the
    collapses from Observation~\ref{obs:collapsing_of_triangle} on $K'$
    obtaining a complex $K_1'$. Then $K_1'$ is a subdivision of $K_1$. Thus, it
    collapses by induction. The other option is that the first elementary
    collapse removes some vertex $a$ and some edge $ab$. Then we obtain a
    subdivision $K_1'$ of $K_1$ by collapses on $K'$ removing $a$ and the
    subdivided edge $ab$ in direction from $a$ towards $b$.

  Now, we show that if $K'$ is collapsible, then $K$ is collapsible again by
    induction on the number of simplices of $K$. Assume that $K'$ is
    collapsible. This implies that $K'$ contains a free
    face $\sigma'$ (a vertex or an edge) which subdivides a face $\sigma$ of
    $K$ which again has to be free. We perform a collapse on $K$ through
    $\sigma$ obtaining $K_1$. As in the previous paragraph, we also 
    collapse $K'$ to a subdivision $K_1'$ of $K_1$. By
    Proposition~\ref{p:greedy_collapses} we get that $K_1'$ is collapsible.
    Therefore, $K_1$ is collapsible by induction which also implies that $K$ is
    collapsible.
  \end{proof}

  \begin{proof}[Proof of Proposition~\ref{p:polynomial}]
    Let $K$ be the input triangulation of $P$. By definition, $\plgcat(P) = 1$
    if and only if $P$ is PL collapsible which occurs if and only if some
    subdivision $K'$ of $K$ is collapsible. By
    Lemma~\ref{lem:collapsible_subdivision}, it is sufficient to check whether
    $K$ is collapsible. 
    This can be done in polynomial time due to Proposition~\ref{p:greedy_collapses}.
  \end{proof}

\section{NP hardness of PL geometric category 2}
\label{s:NP-hardness}

  In this section we prove Theorem~\ref{t:main}. As we have sketched in the
  introduction, in our construction we need to attach a torus to every triangle
  of a certain intermediate complex. We start with the details regarding this
  attachment.
  
  \subsection{Attaching tori}
  \label{ss:attaching_tori}

First, let us us consider the standard torus $T = S^1 \times S^1$. 
An important curve in $T$ is the \emph{longitude} $\lambda = S^1 \times
\{\cdot\}$ where `$\cdot$' stands for some fixed point in $S^1$.

  \begin{definition}[Enriched complex $K^+$]\label{def:enriched_complex}
  Given a simplicial complex $K$, we define the \emph{enriched} complex $K^+$ as
  follows. For each triangle $\tau \in K$ we consider a copy $T_\tau$ of the
  standard torus with longitude $\lambda_\tau$ triangulated as in
  Figure~\ref{f:attached_torus}. We get $K^+$ as a result of gluing all tori $T_\tau$ to
  $K$ so that we identify $\lambda_\tau$ with $\partial\tau$. In the sequel, we
  consider $K$ as well as all the tori $T_\tau$ as subcomplexes of $K^+$.
\end{definition}

Note that the enriched complex $K^+$ can be constructed in polynomial time in the size of $K$.

\begin{figure}
  \begin{center}
    \includegraphics{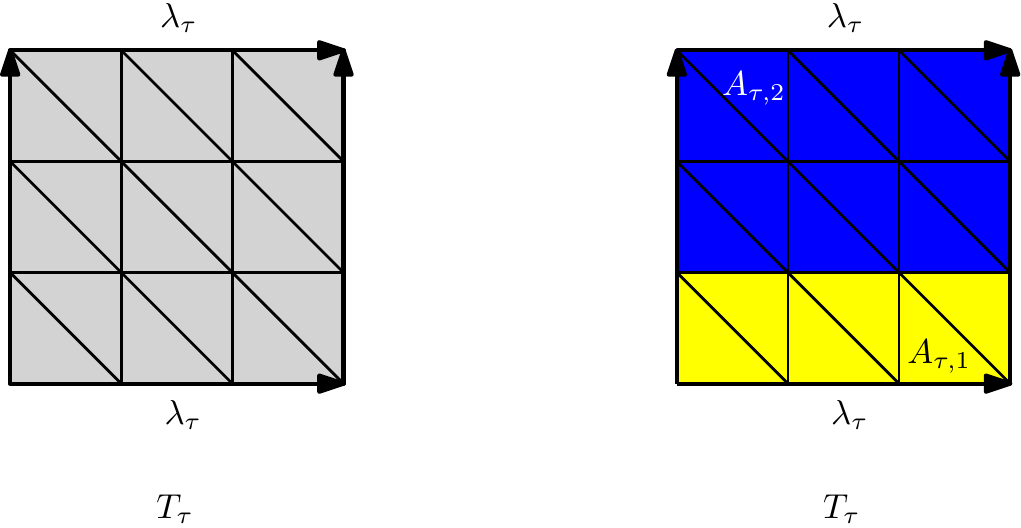}
    \caption{Left: The torus $T_\tau$ with longitude $\lambda_\tau$. Opposite
    edges are identified as usual. Right: Splitting $T_\tau$ to two annuli.}
    \label{f:attached_torus}
  \end{center}
\end{figure}

\begin{observation}\label{obs:attached_torus_plgcat}
  If $K$ admits a covering by two collapsible subcomplexes $K_1, K_2$
  such that both $K_1$ and $K_2$ contain the whole 1-skeleton of $K$ then
  $K^+$ can be also covered by two collapsible subcomplexes.
\end{observation}

\begin{proof}
  Split each $T_\tau$ to two annuli $A_{\tau,1}$ and $A_{\tau,2}$ as in
  Figure~\ref{f:attached_torus}. (Both of them are subcomplexes of $T_\tau$ and they
  share $\lambda_\tau$ on one of their boundaries.) Take $K_i^+$ as the union
  of $K_i$ and all annuli $A_{\tau,i}$ for $i \in \{1,2\}$. Then $K_1^+$ and $K_2^+$
  cover $K^+$. In addition, they are both collapsible because $K_i^+$ collapses
  to $K_i$ as each $A_{\tau,i}$ collapses to $\lambda_\tau$.
\end{proof}

We continue with the main technical lemma for our reduction.

\begin{lemma}\label{lem:attached_torus}
 Let $P$ be a polyhedron which is a union of two subpolyhedra $R$ and $T$.
  Assume that $T = S^1 \times S^1$ is the torus and assume that $R$ and $T$
  intersect exactly in the longitude $\lambda = S^1 \times \{\cdot\}$ of $T$.
  Assume that $P$ can be covered by two contractible subpolyhedra $Q_1, Q_2$.
  Then $\lambda \subseteq Q_1, Q_2$ and $\lambda$ is nullhomologous in $R \cap
  Q_1$ as well as in $R \cap Q_2$.
\end{lemma}

\begin{proof}
Let $A_i := T \cap Q_i$ for $i = 1,2$. The lemma is implied by the following
  two claims where all the homology is considered with $\Z_2$ coefficients. 

\begin{claim}
\label{c:dim2}
\leavevmode
  \begin{enumerate}[(i)]
    \item If $H_1(A_1) = 0$, then $\dim H_1(A_2) \geq 2$.
    \item If $H_1(A_2) = 0$, then $\dim H_1(A_1) \geq 2$.
  \end{enumerate}
\end{claim}

  \begin{claim}
\label{c:dim1}
\leavevmode
    \begin{enumerate}[(i)]
      \item If $H_1(A_1) \neq 0$, then $\dim H_1(A_1) = 1$, $\lambda$ belongs
	to $Q_1$ and $\lambda$ is nullhomologous in $R \cap Q_1$.
      \item If $H_1(A_2) \neq 0$, then $\dim H_1(A_2) = 1$, $\lambda$ belongs
	to $Q_2$ and $\lambda$ is nullhomologous in $R \cap Q_2$.
    \end{enumerate}
  \end{claim}

Indeed, the conjunction of the claims implies that only option is that $\dim
  H_1(A_1) = \dim H_1(A_2) = 1$ and thus we can use the conclusions of
  Claim~\ref{c:dim1}. Therefore, it remains to prove the claims. In each of the
  claims, we only prove the first item as the other one is symmetric.

  \begin{proof}[Proof of Claim~\ref{c:dim2}(i)]
Let $N_1$ be the regular neighborhood%
\footnote{In this case $N_1$ is a
  	$2$-manifold with boundary inside $T$ which collapses to $A_1$. For a general
 	 definition of regular neighborhood see~\cite[Chapter 3]{rourke-sanderson82}.}
	of $A_1$ inside $T$, which is homotopy equivalent to $A_1$; see
	Figure~\ref{f:regular_neighborhood}.
	Then $N_1$ is a surface with boundary. Thus we may apply the Lefschetz duality%
    \footnote{Lefschetz duality (see e.g. Theorem 3.43 in \cite{hatcher02})
    over $\Z_2$:
    Let $M$ be an $n$-dimensional compact manifold
    with boundary $N$. 
    Then $H_i(M,N; \Z_2) \cong H^{n-i}(M; \Z_2)$ for every $i$.}
    obtaining
    \begin{equation}
      \label{e:lefschetz}
      H_1(N_1, \partial N_1) \cong H^1(N_1) \cong H_1(N_1) \cong H_1(A_1) = 0 
    \end{equation}
    where the second isomorphism follows from the fact that the homology and
    the cohomology groups are isomorphic over a field.

\begin{figure}
	\begin{center}
	  \includegraphics[page=3]{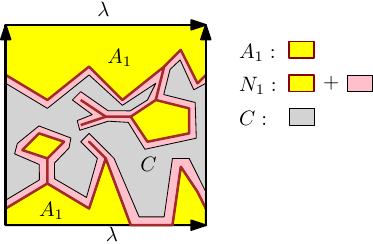}
	  \caption{$A_1$, $N_1$ and $C$ inside $T$.}
	  \label{f:regular_neighborhood}
	\end{center}
\end{figure}

    Now let $C$ be the closure of the complement of $N_1$ in $T$, that is,
    $C := \overline{T \setminus N_1}$. By the excision property of homology,
    and then by~\eqref{e:lefschetz}
    \begin{equation}
    \label{e:excision}
      H_1(T, C) \cong H_1(N_1, \partial N_1) = 0.
    \end{equation}
 
  Finally, we consider the long exact sequence of the pair:


\[
  \cdots \to H_1(C) \xrightarrow[]{i_*} H_1(T) \to H_1(T, C) \to \cdots
\]

 The map $i_*$ is induced by the inclusion $i\colon C \to T$. Because
    of~\eqref{e:excision}, the map $i_*$ is surjective. The inclusion $i$ can
    be decomposed into inclusions $j \colon C \to A_2$ and $k \colon A_2 \to
    T$. (Note that $C \subseteq A_2$ as $A_1$ and $A_2$ cover $T$.) By functoriality of homology, $k_* \colon H_1(A_2) \to H_1(T)$ must be
    surjective as well. Therefore, $\dim H_1(A_2) \geq \dim H_1(T) = 2$.
  \end{proof}

  \begin{proof}[Proof of Claim~\ref{c:dim1}(i)]
  Let 
  $R_1 := R \cap Q_1$. Consider the Mayer-Vietoris exact
  sequence:

  $$
  \cdots \to H_1(A_1 \cap R_1) \xrightarrow[]f H_1(A_1) \oplus H_1(R_1)
  \xrightarrow[]g H_1(Q_1) \to
  \cdots
  $$

  As we assume that $Q_1$ is contractible, we get $H_1(Q_1) = 0$. Therefore, $f$
    is surjective (from exactness). As~we also
  assume that $H_1(A_1) \neq 0$, there is a nonzero vector $v = (z,0) \in H_1(A_1)
  \oplus H_1(R_1)$. We know $v \in \im f$ as $f$ is surjective. In
  particular, $H_1(A_1 \cap R_1) \neq 0$. On the other hand, $A_1 \cap R_1
  \subseteq T \cap R = \lambda$. Therefore, $A_1 \cap R_1 = \lambda$.
  This gives $\lambda \subseteq Q_1$ as we need.
  Using that $f$ is surjective
    again, we get $\dim H_1(A_1) + \dim H_1(R_1) \leq \dim H_1(A_1 \cap R_1) = 1$. Because
    $H_1(A_1) \neq 0$ we actually get $\dim H_1(A_1) = 1$ and $\dim H_1(R_1) =
    0$. This gives that $\lambda$ is nullhomologous in $R_1 = R \cap Q_1$. 
\end{proof}

\end{proof}

  \subsection{Construction form~\cite{goaoc-patak-patakova-tancer-wagner19}}
  \label{ss:cons}

As we sketched in the introduction, we use the construction
from~\cite{goaoc-patak-patakova-tancer-wagner19} as an intermediate step. Given
that this construction is somewhat elaborated, we prefer to state it as a
blackbox only mentioning the properties that we need in our reduction.

The NP-hardness in~\cite{goaoc-patak-patakova-tancer-wagner19} is proved by a
reduction from the classical $3$-satisfiability problem
which is defined in Section~\ref{s:prelim}.

  \begin{proposition}[\cite{goaoc-patak-patakova-tancer-wagner19}]\label{p:construction}
  There is a polynomial time algorithm that produces from a given $3$-CNF
  formula $\phi$ (with $n$ variables) a pure $2$-dimensional complex $K_\phi$ with the following
  properties.

  \begin{enumerate}[(i)]
    \item  	
      $K_\phi$ contains pairwise disjoint triangulated $2$-spheres $S_1,
      \dots, S_n$, one for each variable.
    \item
      The second homology group, $H_2(K_\phi)$, is generated by the spheres
      $S_1, \dots, S_n$. In particular, $H_2(K_\phi) \cong \Z_2^n$ and
      no triangle outside the spheres $S_1, \dots, S_n$ is contained in
      a $2$-cycle.
    \item
      If $\phi$ is satisfiable, then there are triangles $\tau_i$ in $S_i$
      for every $i \in [n]$ such that $K_\phi$ becomes collapsible after
      removing these triangles. In addition, for every $i \in [n]$, there are
      at least two options how to pick $\tau_i$ in $S_i$. (Such a choice can be
      done independently in each $S_i$ yielding at least $2^n$ collapsible
      subcomplexes.)
    \item
      If an arbitrary subdivision of $K_\phi$ becomes collapsible after
      removing some $n$ triangles, then $\phi$ is satisfiable.
  \end{enumerate}

\end{proposition}

\begin{proof} The proof of the proposition consists mostly of references
  to~\cite{goaoc-patak-patakova-tancer-wagner19}. However, a few items are not
  as explicitly stated in~\cite{goaoc-patak-patakova-tancer-wagner19} as we
  need them here, thus we explain in detail how all the items of the
  proposition can be deduced from the text
  in~\cite{goaoc-patak-patakova-tancer-wagner19}.

  The construction of $K_\phi$ is given in Section~4
  of~\cite{goaoc-patak-patakova-tancer-wagner19}. The spheres $S_1, \dots, S_n$
  of item~(i) are the spheres $S(u)$ introduced in \S 4.3
  of~\cite{goaoc-patak-patakova-tancer-wagner19}. For checking the other items,
  we first point out that
  \cite[Proposition~12]{goaoc-patak-patakova-tancer-wagner19} states that the
  number of variables, $n$, is equal to the reduced Euler characteristic
  $\tilde \chi(K_\phi)$.

  It is stated in Remark~13
  in~\cite{goaoc-patak-patakova-tancer-wagner19} that $K_\phi$ is homotopy
  equivalent to the wedge of $n$ $2$-spheres; in particular, $\dim H_2(K_\phi) =
  n$. Then item~(ii) immediately
  follows as the disjoint spheres $S_1, \dots, S_n$ generate a subspace of
  dimension $n$ in $H_2(K_\phi)$. Unfortunately, Remark~13 is only a side
  remark in~\cite{goaoc-patak-patakova-tancer-wagner19} and it is not proved
  there. Therefore, we explain in Appendix~\ref{sec:appendix},
  Proposition~\ref{p:wedge_spheres}, how Remark~13
  of~\cite{goaoc-patak-patakova-tancer-wagner19} follows from their tools.

  Item~(iii), using $n = \tilde \chi(K_\phi)$, is the content of Proposition~8(ii)
  in~\cite{goaoc-patak-patakova-tancer-wagner19} with the addendum that it is
  also necessary to check the proof: In the beginning of Section~7
  of~\cite{goaoc-patak-patakova-tancer-wagner19}, it is specified that the
  triangles are removed in certain regions $D[\ell(u)]$. By checking the
  construction of $D[\ell(u)]$ in \S 4.3
  of~\cite{goaoc-patak-patakova-tancer-wagner19}, these regions are in the correct
  spheres ($S_i$ in our notation; $S(u)$ in the notation
  of~\cite{goaoc-patak-patakova-tancer-wagner19}) and in addition there are at
  least two choices of the removed triangle for every $i$ (actually exactly
  three choices).

  Item~(iv), using $n = \tilde \chi(K_\phi)$, is exactly the content of Proposition~8(iii).
\end{proof}

\subsection{The final reduction}

\begin{proof}[Proof of Theorem~\ref{t:main}]
  Given a 3-CNF formula $\phi$ and its corresponding complex $K_{\phi}$
we construct its enriched complex $K^+_{\phi}$.
  (See Definition~\ref{def:enriched_complex}.) Theorem~\ref{t:main} is
  proved by showing that $\phi$ is satisfiable if and only if $\plgcat(K^+_{\phi})
  \leq 2$ as 3-satisfiability is an NP-hard problem.
\begin{enumerate}[(a)]
  \item \emph{$\phi$ is satisfiable $\implies K^+_{\phi}$ can be covered by two collapsible subcomplexes.}

	  Suppose that the formula $\phi$ is satisfiable. Then
      by Proposition~\ref{p:construction}(iii)
      $K_\phi$ is collapsible after removal of $n$ triangles, one from each sphere $S_i$,
      and for each $S_i$ there are at least two options, say $\tau^{(1)}_{i}, \tau^{(2)}_{i}$,
      how to pick such a triangle. Therefore, the subcomplexes    
   \[
	K_1 := K_\phi \setminus \{\tau^{(1)}_{1}, \dots, \tau^{(1)}_n\}, \,
	K_2 := K_\phi \setminus \{\tau^{(2)}_{1}, \dots, \tau^{(2)}_n\}
   \]
      are collapsible subcomplexes of $K_{\phi}$ and they cover it.

    Moreover, each of $K_1$ and $K_2$ contains the whole 1-skeleton of $K_\phi$.
    Indeed, the complex $K_\phi$ is pure thus every edge of $K_\phi$ is contained in at
    least one triangle and in addition in at least two triangles if it is an
    edge in some of the spheres $S_i$. In order to get $K_1$ or $K_2$, 
    at most one triangle is removed from each $S_i$. Therefore, each edge of
    $K_\phi$ is still contained in at least one triangle of $K_1$ and in at
    least one triangle of $K_2$. Then
    Observation~\ref{obs:attached_torus_plgcat} implies that $K^+_\phi$ can be
    covered by two collapsible subcomplexes. 
     
  \item \emph{A subdivision $\left( K^+_{\phi} \right)^\prime$ of $K^+_{\phi}$
    can be covered by two collapsible subcomplexes $\implies \phi$ is satisfiable.}

First, we sketch the idea: Let $(K^+_1)'$ and $(K^+_2)'$ be the two collapsible subcomplexes of $(K^+_{\phi})'$
    covering it. (We point out that $(K^+_i)'$ is just a notation not implying
    that $(K^+_i)'$ is a subdivision of some complex $K^+_i$.)
    We want to verify the assumption in Proposition~\ref{p:construction}(iv) in
    order to deduce that $\phi$ is satisfiable. For this, we need a subdivision
    of $K_\phi$ such that removing $n$ triangles from this subdivsion yields a
    collapsible complex. In fact, our subdivision will be trivial, thus we need
    to find $n$ triangles in $K_\phi$ such that their removal yields a
    collapsible complex. We will take $(K^+_1)'$, say,
    and we will (essentially) deduce that in each $S_i$ there must
    be $\tau_i$ such that $(K^+_1)'$ must miss at least one triangle in the
    subdivided $\tau_i$. These triangles $\tau_i$ are the triangles we want to
    remove from $K_\phi$. However, we need several intermediate claims to
    deduce that the resulting complex is indeed collapsible. (We will use the
    second complex $(K^+_2)'$ only very sparingly in order to verify the
    assumptions of Lemma~\ref{lem:attached_torus}.)

    Let $K_\phi'$ be the subcomplex of $(K^+_{\phi})'$
    corresponding to $K_\phi$ in this subdivision. (Let us recall that this
    means that $K_\phi'$ is
    formed by simplices $\sigma \in (K^+_{\phi})'$ such that $\sigma \subseteq
    |K_\phi|$.) Let $K'_1 := K_\phi' \cap (K^+_1)'$. 

    \begin{claim}
      \label{c:smaller_collapsible}
The complex $K'_1$ is a collapsible subcomplex of $K_\phi'$.
    \end{claim}

    \begin{proof}

      Our aim is to show that $(K^+_1)'$ collapses to $K'_1$. Then it follows
      from Proposition~\ref{p:greedy_collapses} that $K'_1$ is collapsible.

      We pick an arbitrary triangle $\tau$ of $K_\phi$.
      Recall that $T_\tau$ is the torus attached to $\tau$. (See
      Definition~\ref{def:enriched_complex}.) Let $T'_\tau$ be the
      subcomplex of $(K^+_{\phi})'$ corresponding to $T_\tau$.      
      Note that (the
      subdivsion of) $\partial \tau$ belongs to $(K^+_1)'$ by
      Lemma~\ref{lem:attached_torus}. We also observe that $T'_\tau$ is not a
      subcomplex of $(K^+_1)'$ otherwise $(K^+_1)'$ would contain a nontrivial
      $2$-cycle which is not possible if it is collapsible.
            
      Now we proceed similarly as in the proof of
      Observation~\ref{obs:collapsing_of_triangle}. We greedily perform
      collapses in $(K^+_1)'$ on simplices of $T'_\tau$ with the exception that we are not
      allowed to remove the simplices belonging to (the subdivision of)
      $\partial \tau$. (See Figure~\ref{f:collapse_torus} for a realistic
      example of the intersection of $(K^+_1)'$ and $T'_\tau$.) Let $L'$ be the resulting complex.
      We first observe that $L'$ contains no triangles of $T'_\tau$
      as at least one triangle is missing and the dual graph to our
      triangulation of $T'_\tau$ is connected even after removing the dual edges
      crossing $\partial \tau$. 
      Therefore, $L' \cap T'_\tau$ is a graph. Due
      to our restriction on collapses, subdivided $\partial \tau$ is inside
      this graph. We observe that no other (graph theoretic) cycle may belong to this graph.
      Indeed, another cycle would contain an edge which is not in $\partial
      \tau$, thus not contained in any triangle of $L'$. Therefore, such a
      cycle could not be filled with a $2$-chain, and thus it would be
      necessarily homologically nontrivial in $L'$ which is a contradiction with
      the fact that $L'$ is contractible (obtained by collapses from a
      collapsible complex). Thus, we may conclude that $L' \cap T'_\tau$ is
      the subdivided $\partial \tau$ with a collection of pendant trees. However,
      these pendant trees have to be actually trivial as they get collapsed
      during the greedy collapses.

      \begin{figure}
	\begin{center}
	  \includegraphics[page=2]{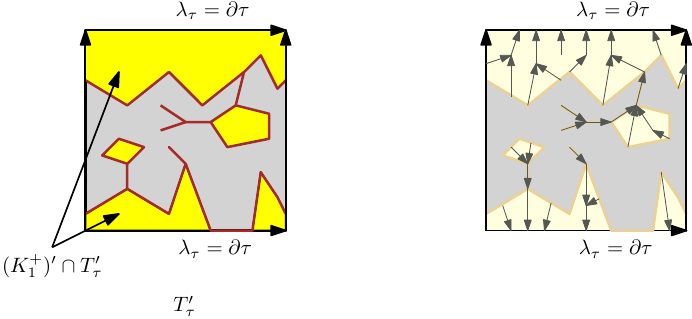}
	  \caption{Left: A realistic example how may $(K_1^+)'$ intersect
	  $T_\tau$. Right: Greedy collapses (only very schematically without
	  emphasizing the triangulation).}
	  \label{f:collapse_torus}
	\end{center}
 \end{figure}

      Altogether we have collapsed $(K^+_1)'$ to a complex $L'$ which agrees
      with $K_1'$ on $K_\phi'$ while we have removed all simplices of $T'_\tau$
      except those that belong to $K_\phi'$. Now we pick another triangle
      $\sigma$ of $K_\phi$ and we remove (via collapses) the simplices of $T'_\sigma$ except
      those belonging to $K_\phi'$ by an analogous approach. After passing
      through every triangle of $K_\phi$, we get exactly $K_1'$ as required.
    \end{proof}

\begin{claim}
  \label{c:smaller_nullhomologous}
  For every triangle $\tau \in K_\phi$, $\partial \tau$ is contained in $|K'_1|$ and it is nullhomologous in
  $|K'_1|$. 
\end{claim}

\begin{proof}
  Let $P := |K^+_\phi| = |(K^+_\phi)'|$. Let $R$ be the polyhedron of $K_\phi$ and all tori of $K^+_\phi$ except
  $T_\tau$. Let $Q_1 := |(K^+_1)'|$ and $Q_2 := |(K^+_2)'|$. Then $R$, $|T_\tau|$,
  $Q_1$ and $Q_2$ satisfy the assumptions of Lemma~\ref{lem:attached_torus}.
  Then we deduce that $\partial \tau$ is nullhomologous in $R \cap Q_1$. Assume
  that 
  $\tau$ is such that $T_\tau'$ is the first torus to be removed in the proof
  of Claim~\ref{c:smaller_collapsible}, we can choose so. Then $R \cap Q_1$ is
  exactly the polyhedron of $L'$ in the proof of
  Claim~\ref{c:smaller_collapsible}. In particular, $L'$ collapses to $K'_1$.
  As collapses provide a homotopy equivalence, we deduce that $\partial \tau$ is nullhomologous in
  $|K'_1|$ as well.
\end{proof}
  
  Now, for any triangle $\tau \in K_\phi$ let $\tau'$ be the subcomplex of
  $K_\phi'$ corresponding to this triangle. 

  \begin{claim}
  \label{c:one_misses}
  \leavevmode
    \begin{enumerate}[(i)]
      \item If $\tau \in K_\phi$ is a triangle which does not belong to any of
	the spheres $S_1, \dots, S_n$, then $\tau'$ is a subcomplex of $K_1'$.
      \item For every $i \in [n]$, all triangles $\tau$ in $S_i$ except
	exactly one satisfy that $\tau'$ is a subcomplex of $K_1'$.
    \end{enumerate}
  \end{claim}

\begin{proof}
  Let $\tau \in K_\phi$. Due to Claim~\ref{c:smaller_nullhomologous}, it has to
  be possible to fill the subdivision of $\partial \tau$ by some $2$-chain $c =
  c(\tau)$ in
  $K_1'$. 

  If $\tau$ does not belong to any of the spheres $S_1, \dots, S_n$, then the
  only option for $c$ is to contain all simplices of $\tau'$. Indeed, if there
  is another such $c'$, then considering $\tau'$ as a $2$-chain, we get a
  nontrivial $2$-cycle $\tau' + c'$ with support at least partially outside the
  spheres $S_1, \dots, S_n$ which contradicts
  Proposition~\ref{p:construction}(ii). Therefore, $\tau'$ must be a subcomplex
  of $K_1'$ which concludes (i).

 Now for (ii), take $i \in [n]$.
Then $K'_1$ has to miss at least one triangle in $|S_i|$ otherwise subdivided
  $S_i$ forms a non-trivial $2$-cycle in $K'_1$ which is a contraction with
  Claim~\ref{c:smaller_collapsible}. Assume that $\tau$ in $S_i$ was chosen so
  that this missing triangle belongs to $\tau'$. Then $\partial \tau$ splits
  (subdivided) $S_i$ to two hemispheres; one of them is formed by $\tau'$ and
  another is formed by the union of subcomplexes $\sigma'$ taken over all
  triangles $\sigma$ in $S_i$ different from $\tau$. By using Proposition~\ref{p:construction}(ii)
  again, the only options are that $c = c(\tau)$ contains all the simplices of one or the
  other (subdivided) hemispheres. But the hemisphere of $\tau'$ is ruled out as
  $\tau'$ misses a triangle of $K'_1$. Thus $c$ has to be filled by the other
  hemisphere. Then we conclude (ii) for all simplices $\sigma$ in $S_i$ except
  exactly $\tau$ as required.
\end{proof}

In the light of Claim~\ref{c:one_misses}(ii), let $\tau_i$ be the unique triangle of $S_i$ such
that $\tau'_i$ is not a subcomplex of $K'_1$. Let $K^-_{\phi}$ be the
subcomplex of $K_\phi$ obtained by removing all triangles $\tau_1, \dots,
\tau_n$ and let $(K^-_{\phi})'$ be the subcomplex of $K'_\phi$ corresponding to
$K^-_{\phi}$. 
Note that
Claim~\ref{c:one_misses} implies that $(K^-_{\phi})'$ is a subcomplex of
$K'_1$. See Figure~\ref{f:comparison} for comparison of $K_\phi'$, $K_1'$ and
$K^-_\phi$ after using Claim~\ref{c:one_misses}.

\begin{figure}
  \begin{center}
    \includegraphics{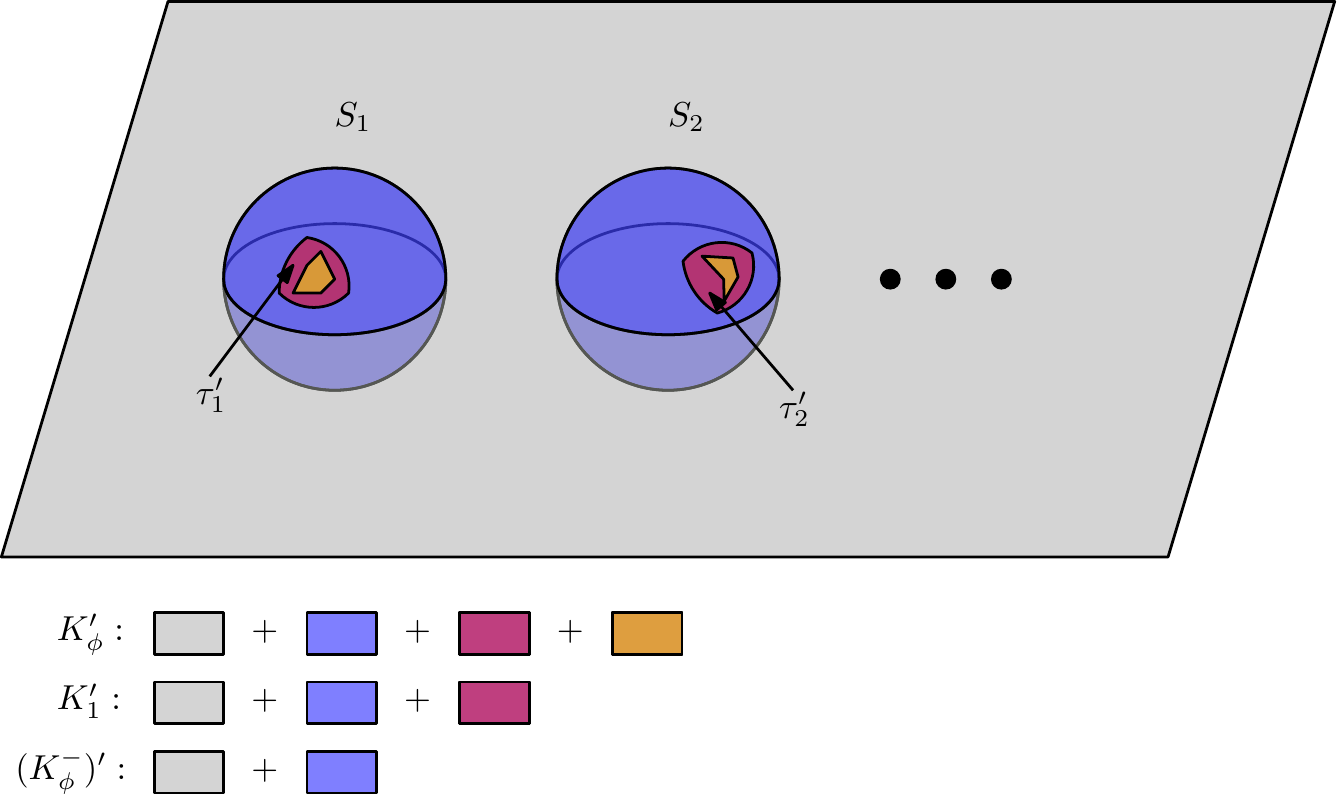}
    \caption{A schematic drawing of $K_\phi'$, $K_1'$ and
$K^-_\phi$. We emphasize that this not really a realistic drawing of $K_\phi'$
    (with the same polyhedron as $K_\phi$)
    as constructed in~\cite{goaoc-patak-patakova-tancer-wagner19}. We only
    attempt to draw as simple complex as possible satisfying conclusions (i)
    and (ii) of Proposition~\ref{p:construction} and so that $K_1'$ is
    collapsible. (The space inside the spheres is completely hollow.)}
    \label{f:comparison}
  \end{center}
\end{figure}

\begin{claim}
\label{c:even_smaller}
  $K'_1$ collapses to $(K^-_\phi)'$.
\end{claim}

\begin{proof}
  The complexes $K'_1$ and $(K^-_\phi)'$ differ only so that $K'_1$ may contain
  some simplices of $\tau'_i$ for some $i$ (except those that subdivide
  $\partial \tau_i$) which are not in $(K^-_\phi)'$.
  
  Now, we continue analogously as in the proof of
  Observation~\ref{obs:collapsing_of_triangle}
  or~Claim~\ref{c:smaller_collapsible}. We greedily collapse all simplices of
  $K'_1$ in $\tau'_i$ except those that subdivide $\partial \tau_i$. We first deduce that
  the resulting complex contains no triangles of $\tau'_i$ as at least one
  triangle was missing in the beginning. Then we deduce that there is no
  graph-theoretic cycle among simplices of $\tau'_i$ except the one
  corresponding to $\partial \tau_i$ by the same argument as in the proof of
  Claim~\ref{c:smaller_collapsible} (using that $K'_1$ is collapsible). Then,
  we deduce that among the simplices of $\tau'_i$ only 
  the simplices subdividing $\partial \tau_i$ remain in the complex. After
  repeating this approach for every $i \in [n]$ we obtain $(K^-_\phi)'$.
\end{proof}

Now, we have acquired enough tools to conclude the case (b) and therefore to conclude the
proof of the theorem. From Claims~\ref{c:smaller_collapsible}
and~\ref{c:even_smaller} and Proposition~\ref{p:greedy_collapses} we deduce
that $(K^-_\phi)'$ is collapsible. By Lemma~\ref{lem:collapsible_subdivision} we
deduce that $K^-_\phi$ is collapsible. Finally, by
Proposition~\ref{p:construction}(iv) we deduce that $\phi$ is satisfiable.
\end{enumerate}
\end{proof}

\section*{Acknowledgments}
We thank to three anonymous referees for providing very useful comments.

\bibliographystyle{alpha} 
\newcommand{\etalchar}[1]{$^{#1}$}

\appendix

\section{Appendix} \label{sec:appendix}

Our aim in the appendix is to verify Remark~13
in~\cite{goaoc-patak-patakova-tancer-wagner19} which is stated but not proved
in~\cite{goaoc-patak-patakova-tancer-wagner19}. The exact statement we need is
given by the following proposition. We will provide all the necessary detail in
order to verify correctness of Remark~13
of~\cite{goaoc-patak-patakova-tancer-wagner19}. On the other hand, we warn the
reader that our proof is not self-contained but it relies on the construction
of $K_\phi$ and partially the notation
in~\cite{goaoc-patak-patakova-tancer-wagner19}; thus it is necessary to consult
the contents of~\cite{goaoc-patak-patakova-tancer-wagner19}.

\begin{proposition}
\label{p:wedge_spheres}
  The complex $K_\phi$ from~\cite{goaoc-patak-patakova-tancer-wagner19} is
  homotopy equivalent to the wedge of $n$ $2$-spheres (where $n$ is the number
  of variables).
\end{proposition}

In the proof, we need the following simple lemma.

\begin{lemma}
\label{l:homotopy}
  Let $K_1, K_2$ be simplicial complexes. Assume that $K_1 \cap K_2$ and $K_2$
  are contractible, then $K_1$ and $K_1 \cup K_2$ are homotopy equivalent.
\end{lemma}

\begin{proof}
  It is well known that contracting a contractible subcomplex is a homotopy
  equivalence~\cite[Proposition~4.1.5]{matousek03}. Therefore, we get 
\[
|K_1 \cup K_2| \simeq |K_1 \cup K_2|/|K_2| = |K_1|/|K_1 \cap K_2| \simeq |K_1|
\]
as required.
\end{proof}

\begin{proof}[Proof of Proposition~\ref{p:wedge_spheres}]
  We follow essentially in verbatim the proof of Proposition~12
  in~\cite{goaoc-patak-patakova-tancer-wagner19}. The only difference is that
  we use Lemma~\ref{l:homotopy} instead of the weaker statement
  in~\cite{goaoc-patak-patakova-tancer-wagner19}: If~$K_1 \cap K_2$ and $K_2$
  are contractible, then $\tilde \chi (K_1 \cup K_2) = \tilde \chi(K_1)$ where
  $\tilde \chi$ stands for the reduced Euler characteristic.

As described in the proof of Proposition~12
  in~\cite{goaoc-patak-patakova-tancer-wagner19}, the complex $K_\phi$ can be
  transformed into certain complex $K'$ by a series of steps when we decompose
  some intermediate complex as $K_1 \cup K_2$ where $K_2$ and $K_1 \cap K_2$
  are contractible, and then we replace the intermediate complex with $K_1$.
  Therefore, using Lemma~\ref{l:homotopy} we get that the resulting complex
  $K'$, after performing all these steps is homotopy equivalent to $K_\phi$.

  By a further homotopy equivalence Goaoc et al.,
  \cite{goaoc-patak-patakova-tancer-wagner19}, obtain another complex $K''$ which is already (obviously) homotopy equivalent
  to the wedge of $n$ $2$-spheres. Therefore, $K_\phi$ is homotopy equivalent to the
  wedge of $n$ $2$-spheres.
\end{proof}

\end{document}